\pdfoutput=1
\RequirePackage{ifpdf}
\ifpdf 
\documentclass[pdftex]{sigma}
\else
\documentclass{sigma}
\fi

\newcommand{\ed}{\mathrm{d}}
\newcommand{\Der}[2]{{{#1}_*}_{#2}}
\newcommand{\e}[1]{\mathrm{e}^{#1}}
\newcommand{\R}{\mathbb{R}}
\newcommand{\Lie}[2]{\operatorname{\mathsterling}_{#1}{#2}}

\numberwithin{equation}{section}

\newtheorem{Theorem}{Theorem}[section]
\newtheorem{Corollary}[Theorem]{Corollary}

\newtheorem{Proposition}[Theorem]{Proposition}
 { \theoremstyle{definition}
\newtheorem{Definition}[Theorem]{Definition}

\newtheorem{Example}[Theorem]{Example}
\newtheorem{Remark}[Theorem]{Remark} }

\begin{document}
\allowdisplaybreaks

\newcommand{\arXivNumber}{1807.00873}

\renewcommand{\PaperNumber}{015}

\FirstPageHeading
	
\ShortArticleName{A Geometric Approach to the Concept of Extensivity in Thermodynamics}
	
\ArticleName{A Geometric Approach to the Concept of Extensivity\\ in Thermodynamics}
	
\Author{Miguel \'Angel GARC\'IA-ARIZA}
	
\AuthorNameForHeading{M.\'A.~Garc\'ia-Ariza}
	
\Address{Instituto de Ciencias, Benem\'erita Universidad Aut\'onoma de Puebla,\\
72750, Puebla, Pue., Mexico}
\Email{\href{mailto:magarciaariza@gmail.com}{magarciaariza@gmail.com}}

\ArticleDates{Received May 24, 2018, in final form February 22, 2019; Published online March 02, 2019}
	
\Abstract{This paper presents a rigorous treatment of the concept of \textit{extensivity} in equilibrium thermodynamics from a geometric point of view. This is achieved by endowing the manifold of equilibrium states of a system with a smooth atlas that is compatible with the pseudogroup of transformations on a vector space that preserve the radial vector field. The resulting geometric structure allows for accurate definitions of extensive differential forms and scaling, and the well-known relationship between both is reproduced. This structure is represented by a global vector field that is locally written as a radial one. The submanifolds that are transversal to it are embedded, and locally defined by extensive functions.}
	
	
\Classification{80A05; 80A10}

\rightline{\textit{To the memory of Salvador Alejandro Ju\'arez-Reyes.}}

\section{Introduction}

The concept of \textit{extensivity} plays a central role in equilibrium thermodynamics. Remarkably, it lacks a precise geometric formulation\footnote{The theory of quasi-homogeneous functions over vector spaces is well established. Manifolds of equilibrium states, in contrast, lack an algebraic structure.}, in spite of the increasing interest in the application of differential geometry to this branch of physics.

A geometric notion of extensive functions, not relying on particular coordinate expressions, is presented herein. This is done by importing the radial vector field $R$ of the Euclidean space to a~smooth finite-dimensional manifold $M$ (the manifold of equilibrium states of a thermodynamic system) via suitable local parameterizations. Thereby, a global vector field $\rho$ that is locally written like an Euler vector field is induced. By means of it and using Euler's theorem, we are able to define extensive functions on~$M$, and then extend this definition to differential forms. The notion of \textit{extensivity} presented here recovers locally all the features that are well known in thermodynamics. Namely, extensive functions are degree-1 homogeneous functions of the extensive variables of a system and extensive differential forms are scaled correspondingly. The submanifolds of $M$ that are transversal to~$\rho$ generalize geometrically the distribution defined by adiabatic hypersurfaces and manifolds of equlibrium states of closed systems, in the sense that they are locally defined by extensive functions.

Many of the basic concepts of equilibrium thermodynamics have long ago been established in a coordinate-free language. For instance, the statement of the second law of thermodynamics was translated to the integrability of certain distribution by Carath\'eodory~\cite{caratheodory1909a}. Another example of this kind is the first law of thermodynamics, which is elegantly formulated in terms of contact geometry~\cite{Hermann1973a}. The more modern geometric approaches to the subject aim to describe the critical points of a thermodynamic system and its underlying microscopic interactions using the scalar curvature of certain metrics~\cite{quevedo2011a, ruppeiner2010a}. A~common feature of all these treatments is that extensive functions are defined as degree-1 homogeneous functions of a preferred global coordinate chart that is known \textit{a~priori} from physical considerations. Assuming its existence forces the topological structures appearing in thermodynamics to be Euclidean. It also leads to further specific problems, as is explained below.

Like extensive functions, extensive differential forms are an important ingredient of the geometric formulation of thermodynamics. These are forms that scale uniformly under the flow of the \textit{Euler vector field} on $M$, given by $Y=x^i\partial_i$, where $x^1,\ldots,x^n$ denote the extensive variables of the system and Einstein's summation convention is used~\cite{belgiorno2003a}. A relevant example of extensive 1-form is the infinitesimal heat of a system, denoted here by $\vartheta$. Its being extensive means that~$Y$ is a symmetry thereof. Upon assuming that~$\vartheta(Y)\neq0$, it follows that~$\vartheta(Y)$ is an integrating factor of~$\vartheta$~\cite{saurel1905a}, provided that it is integrable according to the second law of thermodynamics. Thus, it is proportional to the differential of an extensive function~-- entropy~-- which is unique modulo a constant scale factor. Notice that this result lies upon a coordinate-dependent definition of $Y$. Besides constraining the global topological structure of~$M$, relying on coordinates to define~$Y$ prevents us from regarding entropy as a genuinely geometric object.

Let us consider now the problem of defining extensive functions geometrically under the contact-geometric approach to thermodynamics. The setting of this formalism is the \textit{thermodynamic phase space}, which is a contact manifold $\big(P^{2n+1},\varTheta\big)$ \cite{bravetti2019a, mrugala1991a}. The space of states of a~thermodynamic system is a submanifold $\imath\colon M^n\hookrightarrow P$ satisfying $\imath^*\varTheta=0$. Extensive variables arise as restrictions to $M$ of global Darboux coordinates, i.e., coordinates for which the contact form is written as $\varTheta=\ed w-p_i\ed q^i$. The most appealing feature of this geometric treatment of thermodynamics is that $\varTheta$ is Legendre-invariant. This turns out to be also the main drawback of this framework: since the former mappings lack a coordinate-free description, there is no means to identify the extensive variables in an arbitrary set of Darboux coordinates~\cite{quevedo2007a}.

As was mentioned before, there is a line of research that aims to describe phase transitions by means of curvature. One of these approaches is geometrothermodynamics, which is a~contact-Riemannian-geometric formalism on $P$ that studies critical phenomena and microscopic interactions via the scalar curvature of metrics induced on $M$~\cite{quevedo2007a}. These are defined to be Legendre-invariant metrics compatible with the contact structure of $P$. The former condition yields a whole family of metrics, whose members may be singled out by using a coordinate-dependent notion of \textit{extensivity} on~$M$~\cite{quevedo2017a}. To this end, it is necessary to know the extensive variables of the system \textit{a priori}.

The problem of defining \textit{extensivity} from a coordinate-free point of view has been addressed before in the context of Ruppeiner geometry. Like geometrothermodynamics, this seminal theory also relates scalar curvature to phase transitions and microscopic interactions~\cite{ruppeiner2010a}. The metrics involved are defined by the Hessian of entropy, whose geometric description requires that $M$ be endowed with a flat affine connection~$\nabla$. Demanding that the Christoffel symbols of the latter vanish in the frame induced by extensive variables amounts to defining a global vector field $\varrho$ that forms, together with $\nabla$, a radiant structure on~$M$~\cite{garciaariza2018a}. Motivated by Euler's theorem, extensive functions are defined geometrically by means of $\varrho$.

The present work improves the previous attempts to describe \textit{extensivity} in geometric terms. Like in Ruppeiner geometry, this is done using an appropriate vector field $\rho$ and Euler's theorem, as was said before. The structure involved is less robust than the former, though, since it lacks a connection. Extensive variables are local under this approach, and therefore their existence on $M$ does not restrict its topology to be globally Euclidean.

Under the approach of this paper, the Euler vector field $Y$ mentioned above is the local version of $\rho$. The latter is also a symmetry of infinitesimal heat, and provides an integrating factor for $\vartheta$. Its global coordinate-independent definition allows us to portray thermodynamics as a genuinely geometric theory, following~\cite{belgiorno2010a} (see equation~\eqref{eq:Belgiorno} below). The existence and uniqueness (modulo a constant scale factor) of entropy are a consequence of a feature that \textit{all} manifolds transversal to $\rho$ share: they are locally defined by extensive functions which are unique up to a constant scale factor.

This paper is organized as follows. The main definitions are presented in Section~\ref{sec:ev}. Moreover, it is shown that the notion of \textit{extensivity} provided here agrees with the common one that relies on scaling. The Euler equation also holds in this case. As was mentioned before, endowing $M$ with a suitable structure to describe extensive variables is equivalent to defining a global vector field $\rho$ having locally the form of an Euler vector field. Hence, the question of existence of an extensive structure on $M$ may be translated to the analysis of the singularities of vector fields. We briefly deal with this relationship in Section~\ref{sec:existence}. The submanifolds of $M$ transversal to $\rho$ are studied in Section~\ref{sec:transversal}. Finally, Section~\ref{sec:conclusion} is devoted to concluding remarks.

\section{Definitions and basic results}\label{sec:ev}
In what follows, all vector and tensor fields are assumed to be smooth. We denote by $M$ the finite-dimensional manifold of equilibrium states of a pure, simple thermodynamic system.

We have mentioned before that the usual concept of \textit{extensive functions} in thermodynamics relies on the well-known notion of \textit{homogeneous function} \cite{grudzinski1991a}. Recall that if $V$ is an $n$-dimensional real vector space and $U$ is an open subset of $V$, we say $f\colon U\to\mathbb{R}$ is a \emph{degree-$1$ homogeneous function} if
\begin{gather}\label{eq1ohf}
f(\lambda v)=\lambda f(v),
\end{gather}
for all $v\in U$ and $\lambda\in{}]0,\infty[$ for which $\lambda v\in U$. Immediate examples of degree-1 homogeneous functions are real-valued linear functions on~$V$.

A remarkable feature of equation~\eqref{eq1ohf} is that it does not involve any particular set of coordinates. If we wished to make a similar definition on~$M$, we would require a group action of~$]0,\infty[$ on $M$ whose definition is coordinate independent. Physically, this amounts to describing the scaling of systems without referring to extensive variables. We will circumvent this task and follow an alternative route to \textit{extensivity}, pointed out by a well-known theorem by Euler (see, for instance,~\cite{anosov1997a}), which establishes that smooth degree-1 homogeneous functions may be written in terms of their derivative along the radial vector field $R$ on $V$. For the sake of self-containment, we provide a proof of this result. We remind the reader that, since the tangent bundle of $V$ may be \textit{canonically} identified with $V\times V$, $R$ can be written in a coordinate-free fashion as $R=\mathrm{id}\times\mathrm{id}$, where $\mathrm{id}$ represents the identity mapping. As usual, we denote by $\ed$ the exterior derivative of $k$-forms and by $\mathrm{C}^\infty(U)$ the set of smooth functions defined on an open set~$U$.

\begin{Theorem}[Euler]\label{Theorem:Euler}Let $U\subset V$ be open and $f\in \mathrm{C}^\infty(U)$. Then $f$ is a degree-1 homogeneous function if and only if
\begin{gather}\label{eq:Euler}
\ed f(R)=f.
\end{gather}
\end{Theorem}

\begin{proof}Let $\gamma$ denote the integral curve of $R$ starting at $v$ (this is, $\gamma\colon \R\to V$ is given by $\gamma(t)=\e{t}v$, for all $t\in\mathbb{R}$). If $U$ is an open subset of $V$ and $f\in\mathrm{C}^\infty(U)$ is a degree-1 homogeneous function, then $f\circ\gamma (t)=\e{t}f(v)$ for all $t$ lying in some open interval around $t=0$. Therefore, $\ed f_v(R_v)=\ed(f\circ\gamma)/\ed t|_{t=0}=f(v)$. Since $v$ is an arbitrary element of $V$, we obtain equation~\eqref{eq:Euler}.

Conversely, if $U$ is an open subset of $V$ and $f\in\mathrm{C}^\infty(U)$ satisfies $\ed f(R)=f$, then $\ed (f\circ\gamma)/(\ed t)(t)=f\circ\gamma (t)$ for all $t\in\R$ such that $\gamma(t)\in U$. Integrating the last equation yields $f\circ\gamma(t)=\e{t}f\circ\gamma(0)$. This means that $f\big(\e{t}v\big)=\e{t}f(v)$, for any $v\in V$ and $t\in\mathbb{R}$ satisfying $\e{t}v\in U$, whence $f$ is a degree-1 homogeneous function.
\end{proof}

Henceforth, we shall consider $V=\mathbb{R}^n$. Equation~\eqref{eq:Euler} suggests that pushing $R$ forward to $M$ consistently through local parameterizations might help to define extensive functions thereon. By \textit{consistently}, we mean both that the resulting vector field is globally defined and that it does not depend on the parameterization. To be more precise, if $(U_1,\phi_1)$ and $(U_2,\phi_2)$ are two overlapping charts belonging to the smooth atlas of $M$, then for any $x\in U_1\cap U_2$,
\begin{gather} \label{eq:R}
{{\phi_1}^{-1}_*}_{\phi_1(x)}\left(R_{\phi_1(x)}\right)={{\phi_2}^{-1}_*}_{\phi_2(x)} (R_{\phi_2(x)} )
\end{gather}
must hold (we denote by ${F_*}_p$ the derivative of a mapping $F$ on a point $p$). Equivalently, both the transition function $\psi_{12}=\phi_2\circ\phi_1^{-1}$ and its inverse $\psi_{21}$ must leave $R$ invariant, i.e., $\psi_{12}$ has to satisfy ${{\psi_{12}}_*}_p(R_p)=R_{\psi_{12}(p)}$, for any $p\in\phi_1(U_1\cap U_2)$, and ${{\psi_{21}}_*}_q(R_q)=R_{\psi_{21}(q)}$, for any $q\in\phi_2(U_1\cap U_2)$. These last conditions motivate the following.

\begin{Definition} We refer to diffeomorphisms $F$ defined on an open subset of $\mathbb{R}^n$ satisfying \begin{gather}\label{eq:H}
{F_*}_p(R_p)=R_{F(p)}
\end{gather}
 as \textit{degree-$1$ homogeneous diffeomorphisms}. The set of all such diffeomorphisms will be denoted by $H$.
\end{Definition}

Observe that equation~\eqref{eq:H} can be regarded as the analogue of equation~\eqref{eq:Euler}, provided that~$R_{F(p)}$ may be identified with~$F(p)$. In fact, the similarity goes beyond a mere analogy, as the elements of $H$ behave like degree-1 homogeneous functions under scaling.

\begin{Proposition}
A diffeomorphism $F$ defined on an open set $U$ of the Euclidean space belongs to $H$ if and only if
\begin{gather}\label{eq:d1hd}
F(\lambda p)=\lambda F(p),
\end{gather}
for all $p\in U$ and all $\lambda\in{} ]0,\infty[$ for which $\lambda p\in U$.
\end{Proposition}

\begin{proof}Let $U$ be an open subset of $\mathbb{R}^n$ and $\big(u^1,\ldots,u^n\big)$ denote the cartesian coordinates thereon. By defining $F^i:=u^i\circ F$ for each $i\in\{1,\ldots,n\}$, it can readily be seen that $F\in H$ if and only if every $F^i$ satisfies equation~\eqref{eq:Euler}, whence the result follows.
\end{proof}

Stemming from the proposition above, linear operators on~$\mathbb{R}^n$ are straightforward examples of degree-1 homogeneous diffeomorphisms. The aim of introducing the latter is to state equation~\eqref{eq:R} in the language of atlases compatible with pseudogroups of transformations. This can be achieved owing to the fact below.

\begin{Proposition}The set of degree-$1$ homogeneous diffeomorphisms on the Euclidean space is a group.
\end{Proposition}

\begin{proof}Since the identity mapping $\mathrm{id}$ belongs to $H$, the latter is nonempty.

Suppose that $F\in H$. Then, for any $p\in\R^n$, $\Der{F}{F^{-1}(p)}(R_{F^{-1}(p)})=R_p$. Besides, $R_p=\Der{F}{F^{-1}(p)}\circ\Der{F^{-1}}{p}(R_p)$. The last two expressions imply that $\Der{F^{-1}}{p}(R_p)=R_{F^{-1}(p)}$, whence $F^{-1}\in H$.

Finally, given $F_1,F_2\in H$ and $p\in\R^n$, we have that $\Der{(F_1\circ F_2)}{p}(R_p)=\Der{F_1}{F_2(p)}\allowbreak(\Der{F_2}{p}(R_p))=\Der{F_1}{F_2(p)}(R_{F_2(p)})=R_{F_1\circ F_2(p)}$. Thus, $F_1\circ F_2\in H$, which completes the proof.
\end{proof}

We denote by $\mathcal{H}$ the pseudogroup of transformations on~$\mathbb{R}^n$ formed by restrictions of elements of~$H$ to open subsets of~$\mathbb{R}^n$.

Demanding that $R$ is pushed forward to $M$ consistently by local parameterizations means that the corresponding transition functions must belong to $\mathcal{H}$. In more sophisticated terms, we need to furnish $M$ with an atlas compatible with $\mathcal{H}$. We may readily see that the vector field $\rho$ defined on $M$ as
\begin{gather}\label{eq:defDRM}
\rho_p:=\Der{\phi^{-1}}{\phi(p)}(R_{\phi(p)}),
\end{gather}
for each $p\in M$, is both well and globally defined, provided that $\phi$ corresponds to a chart whose domain contains $p$ and that belongs to the aforementioned atlas.

Using the vector field above and inspired by Theorem~\ref{Theorem:Euler}, we can define extensive functions on~$M$. In what follows, we shall assume that the latter is furnished with an atlas compatible with $\mathcal{H}$, which will be denoted by $\mathcal{A}_\mathcal{H}$.

\begin{Definition}\label{Definition:ef}Let $U$ be an open subset of $M$. We say that $f\in\mathrm{C}^\infty(U)$ is an \textit{extensive function} if $\ed f(\rho)=f$.
\end{Definition}

A straightforward example of extensive functions are the coordinate functions that correspond to charts belonging to $\mathcal{A}_\mathcal{H}$. This follows upon observing that if~$f$ is an extensive function defined on a neighborhood of a point $p\in M$ and $\phi$ is a coordinate transformation corresponding to an element of $\mathcal{A}_\mathcal{H}$ around $p$, then
\begin{gather}\label{eq:ef}
\ed f_p(\rho_p)=\ed\big(f\circ\phi^{-1}\big)_{\phi(p)} (R_{\phi(p)} ).
\end{gather}

The equation above has two important, straightforward consequences. We express the first one in the next proposition.

\begin{Theorem}\label{Theorem:ef}Let $U$ be an open subset of~$M$. A function $f\in\mathrm{C}^\infty(U)$ is extensive if and only if for any chart $(W,\phi)\in\mathcal{A}_{\mathcal{H}}$ with $W\subset U$, $f\circ\phi^{-1}$ is a degree-$1$ homogeneous function on~$\mathbb{R}^n$.
\end{Theorem}

It is worth observing that, if $(W',\psi)\in\mathcal{A}_\mathcal{H}$ is any other chart whose domain overlaps with the above-mentioned $W$, then $f\circ\psi^{-1}$ is also a degree-1 homogeneous function. This follows from writing $f\circ\psi^{-1}$ as $f\circ\phi^{-1}\circ\phi\circ\psi^{-1}$ and applying the chain rule.

The other significant by-product of equation~\eqref{eq:ef} is that $\rho|_U=x^i\partial_i$, provided that $(U,(x^1,\ldots$, $x^n))\in\mathcal{A}_\mathcal{H}$. This means that an extensive function $f$ whose domain overlaps with $U$ is locally written as~$x^i\partial_if$, which prompts the following definition.

\begin{Definition}\label{Definition:ev}An \textit{extensive variable} on $M$ is a coordinate function of a chart belonging to~$\mathcal{A}_\mathcal{H}$. The latter shall be referred to as an \emph{extensive structure} on $M$. The charts belonging to the extensive structure of $M$ are called \emph{extensive charts}, and their domains \textit{extensive domains}. The pair $(M,\mathcal{A}_\mathcal{H})$ is called \emph{extensive manifold}.
\end{Definition}

Notice that the contents of Theorem~\ref{Theorem:ef} and Definition~\ref{Definition:ev} together may be rephrased in the standard terms of equilibrium thermodynamics: \textit{extensive functions are degree-$1$ homogeneous functions of any extensive variables of the system}.

\begin{Remark}We write \textit{any extensive variables} and not \textit{the extensive variables}, because these are defined up to an extensive function. In other words, any non-zero extensive function is itself an extensive variable.
\end{Remark}

Finite-dimensional vector spaces endowed with their standard smooth structures are straightforward examples of extensive manifolds. Equilibrium thermodynamics provides other less tri\-vial instances of the latter, as we illustrate below. In what follows, $\Omega^k(U)$ denotes the set of differential $k$-forms defined on an open set $U$.

\begin{Example}\label{Example:ig}We begin by considering the space of equilibrium states of an ideal gas $M_\text{ig}$. We assume that the internal energy, the volume, and the number of particles of the system, denoted respectively by $U$, $V$, and $N$, comprise a global extensive chart thereon. We regard $\vartheta=\mathrm{d}U+p\ed V-\mu\ed N\in\Omega^1(M_\text{ig})$ and $\rho=U\partial_U+V\partial_V+N\partial_N\in\mathfrak{X}(M_\text{ig})$ as the two objects that define the geometric structure of~$M_\text{ig}$. As is usual, the functions $p$ and $\mu$ are given by $p=cU/V$ and $\mu=-U(cN)^{-1}\big[\ln\big(KU^cVN^{-(c+1)}\big)+c+1\big]$, where $K$ and~$c$ are real, positive constants~\cite{callen1985a}.

We require that the geometric structure defined above on $M_\text{ig}$ be well defined for every equilibrium state. This means that both $\vartheta$ and $\rho$ never vanish, which makes $M_\text{ig}$ diffeomorphic, via the global coordinate mapping $\phi=(U,V,N)$, to the first octant of the Euclidean 3-space, denoted by $O_1$. The manifold of states $M_\text{ig}$ (or, equivalently, $O_1$) is an example of an extensive manifold that is not a vector space.
\end{Example}

We have mentioned that not only functions, but also extensive differential forms are an important concept in thermodynamics. They may be readily defined using $\ed f(\rho)=\Lie{\rho}f$ to extend Definition \ref{Definition:ef} .
\begin{Definition}\label{Definition:edf}
Let $U$ be an open subset of $M$, and $k\in\mathbb{N}$. A differential $k$-form $\omega\in\Omega^k(U)$ is \emph{extensive} if $\Lie{\rho}{\omega}=\omega$.
\end{Definition}

Two straightforward instances of extensive differential forms are those representing \textit{infinitesimal heat} and \textit{infinitesimal work} in thermodynamics. The heat form is of particular importance, due both to its geometric properties and the fact that, according to~\cite{belgiorno2010a}, it is sufficient to determine uniquely a thermodynamic system, as we shall explain in the sequel.

Notice that the notion of extensivity can be readily extended to tensor fields, in general. This idea is useful in the Riemannian geometric approaches to equilibrium thermodynamics as we now explain.

\begin{Example}\label{Example:GTD}A coordinate-free approach to Ruppeiner geometry requires that $M$ be endowed with a symmetric $(2,0)$ tensor $g^\text{R}$ and a flat affine connection $\nabla$ satisfying $\nabla\rho=\mathrm{id}$ and $\nabla_Xg^\text{R}(Y,Z)=\nabla_Yg^\text{R}(X,Z)$, for all local smooth vector fields $X$, $Y$, and $Z$ \cite{garciaariza2018a}. The last equation implies that $g^\text{R}$ is locally written as the Hessian of a smooth function $\varPhi$.

We shall consider that $\varPhi$ above is a thermodymamic potential that is \textit{extensive of degree $\beta$}, this is, we will assume that $\varPhi$ satisfies $\Lie{\rho}{\varPhi}=\beta\varPhi$, for some real number $\beta$. When $\beta=1$, then $\varPhi$ is the potential of an \textit{ordinary} thermodynamic system and we recover Definition \ref{Definition:ef}. Otherwise, the system is called \textit{nonordinary} \cite{quevedo2017a}.

Let $\rho^\flat$ stand for the 1-form $X\mapsto g^\text{R}(\rho,X)$. A straightforward computation yields that, locally, $\rho^\flat=(\beta-1)\ed{\varPhi}$. Consequently, ordinary thermodynamic systems are characterized by having $\rho$ as a null vector of $g^\text{R}$.

Motivated by the definition of Quevedo's Legendre-invariant metrics that describe first-order phase transitions \cite{quevedo2010a}, we shall determine the local form of all tensor fields $g$ conformal to $g^\text{R}$ satisfying
\begin{gather}\label{eq:Quevedo}
\Lie{\rho}{g}=2\beta g.
\end{gather}
We begin by observing that $\Lie{\rho}{g^\text{R}}=\beta g^\text{R}$, whence equation~\eqref{eq:Quevedo} implies that, if $g=\lambda g^\text{R}$, then
\begin{gather}\label{eq:conf}
\Lie{\rho}{\lambda}=\beta\lambda.
\end{gather}

As is usually the case in thermodynamics, let us assume that $\rho$ never vanishes. This means that we can always choose a local chart $\big(U,\big(x^1,\ldots,x^n\big)\big)$ so that $x^1\neq0$ and $g|_U=\nabla\ed{\varPhi}$. If we define $y^i:=x^i/x^1$ for each $i\in\{2,\ldots,n\}$, then $\Lie{\rho}{y^i}=0$, and therefore $\rho|_U=\beta\varPhi\partial_\varPhi$, using $\big(U,\big(\varPhi,y^2,\ldots,y^n\big)\big)$ as coordinates. Thus, equation~\eqref{eq:conf} is written on $U$ simply as $\beta\varPhi\partial_\varPhi\lambda=\beta\lambda$, whose solution is $\lambda=f\varPhi$, where $f$ is a local smooth function satisfying $\Lie{\rho}f=0$. This means that Quevedo's metrics $g^\text{I}$ are the only solution to equation~\eqref{eq:Quevedo} up to an \textit{intensive} ($\rho$-invariant) scale factor. In other words, Quevedo's metrics may be regarded as the simplest nontrivial positive definite~-- or semi-definite, depending on the value of~$\beta$~-- solution of equation~\eqref{eq:Quevedo}.
\end{Example}

Before concluding this section, we consider worth mentioning that extensive structures may be portrayed differently, depending on what feature of extensivity is considered to be the most important. For instance, instead of demanding that the transition functions on $M$ leave the radial vector field invariant, we could have required that they leave the \textit{homogeneity} of functions invariant, i.e., that they map (via the pull-back of functions) degree-1 homogeneous functions to degree-1 homogeneous functions. The next result establishes that the geometric structure that corresponds to this requirement is actually an extensive one.

\begin{Proposition}A diffeomorphism $F\colon \mathbb{R}^n\to\R^n$ is degree-$1$ homogenous if and only if for any degree-$1$ homogeneous function $f$ defined on an open subset of $\,\R^n$, $F^*f$ is a degree-$1$ homo\-ge\-neous function.
\end{Proposition}

\begin{proof}Observe that if $f$ is a real-valued function defined on an open subset of $\R^n$ and $F\colon \R^n\to\R^n$ is a diffeomorphism, then for each $p\in\R^n$,
$\ed(F^*f)_p(R_p) =\ed f_{F(p)}\circ\Der{F}{p}(R_p)$.

Let $f$ be a degree-1 homogeneous function, and suppose that $F\in H$. Then $\ed(F^*f)_p(R_p)=\ed f_{F(p)}(R_{F(p)})=f(F(p))=F^*f(p)$. Therefore, $F^*f$ is a degree-1 homogeneous function.

Conversely, if $F$ pulls back any degree-1 homogeneous function to a degree-1 homogeneous function, then for any such $f$ defined on an open subset of $\R^n$ we have that $\ed f_{F(p)}\circ\Der{F}{p}(R_p)=\ed(F^*f)_p(R_p)=F^*f(p)=f(F(p))=\ed f_{F(p)}(R_{F(p)})$. Hence, the derivative of $f$ at $F(p)$ annihilates $\Der{F}{p}(R_p)-R_{F(p)}$. Since the canonical projections $\varpi^i\colon \R^n\to\R$, with $i\in\{1,\ldots,n\}$, are degree-1 homogeneous functions, we have that $\Der{F}{p}(R_p)-R_{F(p)}\in\bigcap_{i=1}^n\ker\ed\varpi^i_{F(p)}=\{0\}$. This implies that $F\in H$.
\end{proof}

So far, we have not made any reference to the relationship between scaling and extensivity that is established by equation~\eqref{eq1ohf}. As expected, defining an extensive structure on $M$ provides a means to define \textit{scaling of equilibrium states} geometrically. Indeed, let $\varphi_t$ denote the (local) flow of $\rho$. For each $p\in M$, the integral curve of $\rho$ starting at~$p$, $\gamma(t):=\varphi_t(p)$, is defined on an open neighborhood of $t=0$, which may be written as $]{-}\varepsilon_p, \varepsilon_p[$ for some $\varepsilon_p>0$ that varies pointwise. We let $\lambda\in{}]\e{-\varepsilon_p},\e{\varepsilon_p}[$ and define $\lambda p:=\varphi_{\log\lambda}(p)$. This operation is not exactly an action of the positive real numbers on $M$. However, it satisfies the familiar properties of uniform scaling on the Euclidean space, and reproduces on $M$ the well-known relationship between scaling an extensivity.

\begin{Proposition}Let $U$ be an open subset of $M$ and $\omega\in\Omega^k(U)$, with $k\in\mathbb{N}\cup\{0\}$. Then $\omega$ is extensive if and only if
\begin{gather}\label{eq:ex4}
\varphi_t^*\omega=\e{t}\omega,
\end{gather}
for every value of $t$ for which equation~\eqref{eq:ex4} makes sense.
\end{Proposition}

\begin{proof}We first show that differential forms satisfying equation~\eqref{eq:ex4} are extensive. Suppose that $\omega\in\Omega^k(U)$ is such a differential form, defined on an open subset $U$ of $M$. Then, there is an open interval $I$ containing $t=0$ such that $(\varphi_t^*\omega-\omega)/t=\omega(\e{t}-1)/t$, for every $t\in I\setminus\{0\}$. This approaches $\omega$ as $t\to0$, i.e., $\Lie{\rho}{\omega}=\omega$.

Conversely, let $\omega\in\Omega^k(U)$. Recall that
\begin{gather}\label{eq:Lie}
\frac{\ed}{\ed t}(\varphi_t^*\omega)=\varphi_t^*\left(\Lie{\rho}{\omega}\right).
\end{gather}
Suppose now that $\omega$ is extensive. Then, for any (time-independent) $X_1,\ldots,X_k\in\mathfrak{X}(M)$, equation~\eqref{eq:Lie} is written as
\begin{gather}\label{eq:exdf}
\frac{\ed}{\ed t}\left[\left(\varphi_t^*\omega\right)(X_1\wedge\cdots\wedge X_k)\right]=(\varphi_t^*\omega)(X_1\wedge\cdots\wedge X_k).
\end{gather}
We have thus that $(\varphi_t^*\omega) (X_1\wedge\cdots\wedge X_k)=\e{t}\omega\big(X^1\wedge\cdots\wedge X_k\big)$, for values of~$t$ around $t=0$ for which~$\varphi_t^*\omega$ is defined. Since the vector fields $X_1,\ldots,X_k$ are arbitrary, equation~\eqref{eq:ex4} follows.
\end{proof}

When $k=0$, equation~\eqref{eq:ex4} reads precisely $f(\lambda p)=\lambda f(p)$, for all $\lambda\in{}]\e{-\varepsilon_p},\e{\varepsilon_p}[$ and $p\in M$ (cf.\ equation~\eqref{eq1ohf}).

The notion of \textit{uniform scaling} of states may also be taken as starting point to define an extensive structure on $M$. Namely, for any $x\in M$ we may intuitively define $\lambda x$ as $\phi^{-1}(\lambda \phi(x))$, where $(U,\phi)$ is a smooth chart whose domain contains $x$. Demanding that this definition be coordinate-independent amounts to requiring the transition functions on $M$ satisfy equation~\eqref{eq:d1hd} for all points on the Euclidean space and all values of~$\lambda$ for which the latter makes sense. Hence, an extensive structure on $M$ may be regarded as a smooth atlas whose charts preserve locally dilations on $\mathbb{R}^n$.

We have seen so far that any extensive manifold is endowed with a global vector field that has locally the form of a radial vector field. In the next section, we will show that actually such vectors embody extensive structures. Therefore, studying the conditions under which a~manifold accepts an extensive structure can be translated to questions regarding the existence of the aforementioned vector fields.

\section{Existence of extensive structures}\label{sec:existence}

This section is devoted to the following question: what kinds of manifolds may be endowed with an extensive structure? We provide a partial answer by means of identifying extensive structures with global vector fields.

As is established in equation~\eqref{eq:defDRM}, any extensive structure defines a global vector field that is locally written like a radial vector field. It is natural to ask whether any such vector field defines an extensive structure. The answer is in the affirmative, as we now show.

\begin{Proposition}\label{Proposition:3}If $M$ is endowed with a vector field $X$ and a smooth subatlas comprising charts where $X$ has the form of a radial vector field, then $M$ is furnished with an extensive structure.
\end{Proposition}

\begin{proof}The result above follows upon observing that $X\in\mathfrak{X}(M)$ has the form of a radial vector field, if and only if for all $p\in M$, there exists a smooth chart $(U,\phi)$ with $p\in U$, that satisfies $X_q=\Der{\phi^{-1}}{\phi(q)}(R_{\phi(q)})$, for all $q\in U$. If $(U',\psi)$ is another smooth chart whose domain contains $p$ and such that $X_q=\Der{\psi^{-1}}{\psi(q)}(R_{\psi(q)})$, for all $q\in U'$, then the corresponding transition function belongs to the pseudogroup $\mathcal{H}$. Indeed, for all $q\in U\cap U'$, $\Der{\phi^{-1}}{\phi(q)}(R_{\phi(q)})=\Der{\psi^{-1}}{\psi(q)}(R_{\psi(q)})$. Upon applying $\Der{\psi}{q}$ to both hand sides of the last equation, we obtain that $\Der{(\psi\circ\phi^{-1})}{\phi(p)}(R_{\phi(p)})=R_{\psi(p)}$, whence $\psi\circ\phi^{-1}\in\mathcal{H}$.

As a consequence, the set of all smooth charts on which $X$ is written as a radial vector field forms an atlas compatible with the pseudogroup $\mathcal{H}$. This atlas is contained in a maximal atlas compatible with $\mathcal{H}$, which yields the desired result.
\end{proof}

We have thus proven that extensive structures are equivalent to global vector fields that are locally written as radial vector fields, which we call \emph{locally-radial vector fields}. Hence, any manifold admitting a locally-radial vector field admits an extensive structure.

A particular instance of locally-radial vector field is a non-vanishing one, as we now show.

\begin{Proposition}\label{Proposition:VR}Let $X\in\mathfrak{X}(M)$. If $p\in M$ is such that $X_p\neq0$, then $X|_U=\Der{\phi^{-1}}{\phi(p)}(R_{\phi(p)})$, for some smooth chart $(U,\phi)$ around $p$.
\end{Proposition}

\begin{proof}Let $p\in M$ be such that $X_p\neq 0$. Then, there exists a chart $\big(U',\big(y^1,\ldots,y^n\big)\big)$ around~$p$ such that $X|_{U'}=\partial_1$. We wish to show that there exist $n$ independent smooth functions $x^1,\ldots,x^n\in\mathrm{C}^\infty(U)$, with $U\subset U'$, satisfying $\ed x^i(X)=x^i$. Because of the form that $X$ has on $U'$, the last expression is equivalent to $\partial x^i/\partial y^1=x^i$, whose general solution is given by $x^i=\e{y^1}G^i$, for all $i\in\{1,\ldots,n\}$. In the last expression, $G^1,\ldots,G^n\in \mathrm{C}^\infty (U)$, for some open set $U$ contained in $U'$. Furthermore, each function $G^i$ satisfies $\partial G^i/\partial y^1 =0$. Setting $x^1:=\mathrm{e}^{y^1}$, and $x^i:=y^i\mathrm{e}^{y^1}$, for all $i\in\{2,\ldots,n\}$ yields a coordinate chart $(U,\phi)$ around~$p$, with $\phi=\big(x^1,\ldots,x^n\big)$. The latter satisfies $X_p=\Der{\phi^{-1}}{\phi(p)}(R_{\phi(p)})$. The result then follows from Proposition~\ref{Proposition:3}.
\end{proof}

As we mentioned before, a direct consequence of Proposition~\ref{Proposition:VR} is the following.

\begin{Corollary}If $M$ admits a non-vanishing vector field then $M$ admits an extensive structure.
\end{Corollary}

It is evident though that being endowed with a non-vanishing vector field is not a necessary condition for a manifold to possess an extensive structure. A straightforward illustration of this claim is the Euclidean space: its radial vector field vanishes at the origin. Yet more, in general, the vector field $\rho$ on an extensive manifold $(M,\mathcal{A}_\mathcal{H})$ may contain countably many singularities. This is because $p$ is a singularity of $\rho$ (i.e., $\rho_p=0$) if and only if an extensive coordinate chart (and hence, all of them containing~$p$) is centered at $p$. Since coordinate mappings are diffeomorphisms, each extensive domain may contain only one singularity of $\rho$, which implies that the set of singularities of $\rho$ is discrete. Because~$M$ is second-countable, $\rho$ has countably many singularities.

According to the previous paragraph, we might think that a global vector field on $M$ with a discrete set of singularities yields an extensive structure on~$M$ (if that were the case, any manifold would admit an extensive structure). Nonetheless, this turns out to be false, as we make evident in the next example.

\begin{Example}Consider the Euclidean plane $\R^2$ with its canonical linear and smooth structures. We define $X\in\mathfrak{X}\big(\R^2\big)$ as $X=-x\partial_y+y\partial_x$, where $(x,y)$ are the cartesian coordinates on the plane and the symbols $\partial_x$ and $\partial_y$ denote the vector fields of the holonomic frame thereby induced.

Let $o$ denote the point $(0,0)\in\mathbb{R}^2$. According to Proposition~\ref{Proposition:VR}, there must exist a coordinate chart $(U,(w,z))$ around each $p\neq o$, such that $X|_U=w\partial_w+z\partial_z$. Indeed, if $\theta$ and $r$ denote the polar coordinates on the plane, then $w=\e{\theta}$ and $z=r\e{\theta}$ are extensive coordinate functions of the extensive structure that $X$ defines on $\R^2\setminus\{o\}$.

Nevertheless, we claim that it is impossible to construct similar functions around $o$. This is the case because the integral curves of $X$ are circles, whereas those of a radial vector field are lines.
\end{Example}

The example above provides some information about the local structure of the flow of a vector field around a singularity, if this vector is to define an extensive structure on a manifold. As we have observed, $X$ may have countably many singularities, but these must be of a particular kind. Studying the conditions over a vector field so that it is locally radial around a singularity is a path to answering the question of existence of extensive structures, and shall be the topic of future work.

We return to the main subject of this paper in the next section, where we deal in general terms with an important geometric property of the heat 1-form in the context of extensive structures.

\section{Submanifolds transversal to the extensive structure}\label{sec:transversal}

We begin this section by pointing out a well-known class of manifolds that are extensive. Recall that a manifold is \textit{affine} if it is endowed with an atlas compatible with the pseudogroup of affine transformations on the Euclidean space \cite{Goldman1981a}. If the coordinate transformations on this manifold are further restricted to be linear mappings, the resulting geometric structure thereon is called \textit{radiant} \cite{Goldman1984a}. The latter turns out to be relevant in the context of thermodynamics, since it provides an appropriate setting for a rigorous description of Ruppeiner geometry \cite{garciaariza2018a} (see Example~\ref{Example:GTD}). Because every linear transformation is a degree-1 homogeneous diffeomorphism, radiant manifolds are examples of extensive manifolds. Obviously, every radiant manifold is endowed with a locally-radial vector field $\rho$. The immersed submanifolds that are transversal to~$\rho$ are locally defined by extensive functions (and are therefore embedded). In this section, we show that the same holds for extensive manifolds in general.

Before proving the above-mentioned result, let us briefly discuss its importance in the context of equilibrium thermodynamics.

As we have said, the manifold of states of a thermodynamic system in equilibrium is a smooth $n$-dimensional extensive manifold. Different thermodynamic systems may share a common space of states (including its extensive structure), as is the case of hydrostatic systems, for instance. The difference between one system and another (e.g., an ideal gas and a van der Waals gas) lies in the so-called \emph{fundamental equation}. This is a coordinate expression for a \emph{thermodynamic potential} $\varPhi\in \mathrm{C}^\infty(M)$.

Two basic thermodynamic potentials (from which any other can be derived via a Legendre transform) are the \emph{internal energy} of the system and its \emph{entropy}. The knowledge of any of these two determines uniquely a thermodynamic system. It turns out that also the infinitesimal heat $\vartheta$ and the locally-radial vector field $\rho$ specify a thermodynamic system. Indeed, since $\vartheta$ is extensive, integrable, and has $\rho$ as a transversal symmetry, then \cite{belgiorno2010a, saurel1905a}
\begin{gather}\label{eq:Belgiorno}
\frac{\vartheta}{\vartheta(\rho)}=\ed \ln S
\end{gather}
for some function local function $S$. The latter is extensive, as can be readily seen upon evaluating both hand sides on $\rho$. Equation~\eqref{eq:Belgiorno} has two important consequences. First, it establishes that thermodynamics is fully determined by the geometry of $M$, considering that $\vartheta$ and $\rho$ comprise it. The second one is that the adiabatic hypersurfaces of~$M$ (the integral manifolds of $\vartheta$) are locally defined by an extensive function. This actually holds for \textit{any} manifold that is transversal to $\rho$. In order to prove this, let us recall that if $\imath\colon N\hookrightarrow M$ is a smooth embedded submanifold of~$M$, a smooth function $f$ defined on an open subset~$U$ of~$M$ is a~\emph{local defining function} for~$N$ if $U\cap\imath(N)$ is a regular level set of $f$, this is, if $U\cap \imath(N)=f^{-1}(c)$, for some regular value~$c$ of~$f$~\cite{lee2013a}. We say that $N$ is \emph{locally defined by extensive functions} if $N$ admits an extensive function as a~local defining function in a neighborhood of each of its points.

\begin{Theorem}A submanifold of $M$ containing no singular points of $\rho$ is transversal to $\rho$ if and only if it is locally defined by nonvanishing extensive functions.
\end{Theorem}

\begin{proof}We prove first that if $U$ is an open subset of $M$ and $f\in \mathrm{C}^\infty(U)$ is extensive, then $f^{-1}(c)$ is transversal to $\rho$, provided that $c\neq0$ is a regular value of $f$. In order to do this, it suffices to show that the regular level sets of extensive variables are transversal to $\rho$, since an extensive function is an extensive variable in a neighborhood of any $p\in U$ such that $\ed f_p\neq0$. Thus, we let $\big(U,\big(x^1,\ldots,x^n\big)\big)$ be an extensive chart and $c$ be a regular value of~$x^1$. If we denote by~$\imath$ the inclusion of $\big(x^1\big)^{-1}(c)$ into $M$, then $\ed\big(\imath^* x^1\big)=0$. This means that for any $p\in\big(x^1\big)^{-1}(c)$ and $v\in\mathrm{T}_p \big(x^1\big)^{-1}(c)$, $\Der{\imath}{p}(v)=a^2{\partial_2}_{\imath(p)}+\cdots+a^n{\partial_n}_{\imath(p)}$. The only common element of $\Der{\imath}{p}(\mathrm{T}_p\big(x^1\big)^{-1}(c))$ and the span of $\rho_{\imath(p)}$ is zero. Indeed, if $\Der{\imath}{p}v=a\rho_{\imath(p)}$, for some $a\in\R$, $0=\ed\big(\imath^*x^1\big)_p(v)=a\ed x^1_{\imath(p)}(\rho_{\imath(p)})=a\imath^*x^1(p)=ac$, which implies that $a=0$. Hence, the tangent space to~$M$ at~$p$, $\mathrm{T}_pM$, may be written as the direct sum of $\Der{\imath}{p}\big(\mathrm{T}_p \big(x^1\big)^{-1}(c)\big)$ and the span of~$\rho_{\imath(p)}$, for any $p\in \big(x^1\big)^{-1}(c)$, and thus $\big(x^1\big)^{-1}(c)$ is transversal to~$\rho$.

In consequence, the regular level hypersurfaces (submanifolds of codimension 1) of extensive functions are transversal to $\rho$. Hence, if $\imath\colon N\hookrightarrow M$ is a manifold locally defined by nonvanishing extensive functions, it is transversal to $\rho$.

Conversely, suppose that $N$ is transversal to $\rho$ and let $p\in N$. Since $\rho_{\imath(p)}\neq0$, there exists a~chart $\big(U,\big(y^1,\ldots,y^n\big)\big)$ around $\imath(p)$ such that $\rho|_U=\partial_1$, where $\imath$ denotes the inclusion of $N$ into $M$. The transversality of $N$ to $\rho$ implies that $\ed (\imath^*y^1)_p=0$. Hence, $\big(U,\big(y^1,\ldots,y^n\big)\big)$ is a slice chart for $N$ around~$p$, and because the existence of such a chart is guaranteed for any point of~$N$, it follows that it is embedded. We define $f\in C^\infty(U)$ as $f:=\e{y^1}$, which is both nonvanishing and extensive: $\ed f(\rho)=f\ed y^1(\rho)=f$. Furthermore, $\ed (\imath^*f)_p=0$, meaning that $\imath^*f$ is constant, i.e., $U\cap\imath(N)=f^{-1}(c)$, for some nonzero $c\in\R$. The latter is a regular value of $f$, since $\ed y^1_{\imath(p)}\neq 0$, as follows from $\ed y^1_{\imath(p)}(\rho_{\imath(p)})=1$. We have thus proven that $N$ is locally defined by nonvanishing extensive functions, as desired.
\end{proof}

The function $S$ in equation~\eqref{eq:Belgiorno} is unique up to a constant scale factor. The same holds for the locally-defining functions of the submanifolds of $M$ transversal to $\rho$.

\begin{Proposition}The local defining functions of manifolds transversal to $\rho$ are unique modulo scale.
\end{Proposition}

\begin{proof}Let $\imath\colon N\hookrightarrow M$ be an immersed submanifold that is transversal $\rho$. Let $f\in\mathrm{C}^\infty(U)$ be a local defining function for $\imath$ around a point $p\in N$, and suppose that $\varPhi\in\mathrm{C}^\infty(U)$ is another local defining function for $\imath$ around $p$. Then, $\ed f\propto\ed\varPhi$, which means that $\ed f\wedge\ed\varPhi=0$. Thus, $\varPhi$ may be written as a function of $f$, this is, there exists a real-valued function $\hat\varPhi$ defined on an open interval $I$ containing the image of $f$ such that $\varPhi=\hat\varPhi\circ f$. Since $\varPhi$ is extensive, we have that $\hat\varPhi$ must satisfy
\begin{gather}\label{eq:uS}
\hat\varPhi(t)=t\hat\varPhi'(t),
\end{gather}
for every $t$ lying in the image of $f$, where $\hat\varPhi'$ denotes the derivative of $\hat\varPhi$. The solution to equation~\eqref{eq:uS} in an open interval containing the image of $f$ is $\varPhi(t)=kt$, with $t\in I$. This means that $\varPhi=kf$, as we wished to prove.
\end{proof}

We conclude this section by noting that the three conditions imposed upon the heat 1-form of a thermodynamic system~-- integrability, extensivity, and transversality~-- are independent from each other.

We begin by observing that forms that are transversal to $\rho$ are not necessarily extensive. Given an extensive 1-form $\alpha$ that is transversal to $\rho$, a straightforward example of a 1-form that is transversal to $\rho$ but is not extensive is $f\alpha$, where $f$ is an extensive function. If the manifold in question is 2-dimensional, then $f\alpha$ is integrable, whence transversality and integrability do not guarantee extensivity.

It is also true that extensivity is not a sufficient condition for transversality, as we illustrate below.

\begin{Example}Consider the set of points $(x,y)$ in the Euclidean plane with both $x>0$ and $y>0$, which we denote by $\mathbb{R}^2_+$, furnished with the smooth and extensive structures that the Euclidean plane induces thereon.

The global 1-form $\alpha:=(1+y/x)\ed x-(1+x/y)\ed y$ is extensive, yet $\alpha(\rho)=0$, whence it is not transversal to $\rho$.
\end{Example}

Notice that the 1-form $\alpha$ of the example above is integrable. This shows that not even integrable extensive forms are necessarily transversal to the extensive structure of a manifold. In brief words, extensivity and integrability do not imply transversality.

Likewise, 1-forms transversal to $\rho$ are not necessarily integrable, as we now show.

\begin{Example}The infinitesimal work on a thermodynamic system is represented by an extensive 1-form $\varepsilon\in\Omega^1(M)$. According to the First Law of thermodynamics, the 1-form $\vartheta+\varepsilon$ is closed, and a local potential is the internal energy~$U$ of the system. A system is called \textit{mechanically conservative} if $\imath^*\ed{\varepsilon}=0$, for any integral submanifold $\imath\colon \varSigma\hookrightarrow M$ of $\vartheta$ (cf.~\cite{edelen2005a}).

Unlike infinitesimal heat, infinitesimal work is not always integrable. From the first law and the second law, we have that $\varepsilon\wedge\ed{\varepsilon}= \ed{\vartheta}\wedge\ed{U}$, which is not zero in general. For instance, $\varepsilon\wedge\ed{\varepsilon}=0$ for an ideal gas, whereas $\varepsilon\wedge\ed{\varepsilon}=aS/\big(cRV^2\big)+[cN(b-V/N)]^{-1}U/N\ed{U}\wedge\ed{V}\wedge\ed{N}$ for a van der Waals fluid. This is readily obtained using the first two laws of thermodynamics and the corresponding fundamental equations: $S=NR\ln\big(K_1U^cV/N^{c+1}\big)$ for the ideal gas, and $S=NR\ln\big[K_2(V/N-b)(U/N+Na/V)^c\big]$ for the van der Waals fluid, where $a$, $b$, $c$, $K_1$, $K_2$, and $R$ are positive constants~\cite{callen1985a}. As is customary, the volume $V$ and the number of particles $N$, together with $U$, induce an extensive structure on the manifold of equilibrium states (excluding the coexistence region in the case of a van der Waals fluid).
\end{Example}

Observe that infinitesimal work is transversal to the extensive structure for both ideal and van der Waals gases. This means that neither transversality nor \textit{extensivity} imply integrability.

\section{Concluding remarks}\label{sec:conclusion}

In several instances, differential geometry seems to be a powerful tool for the study of equilibrium thermodynamics. The aim of this paper was to contribute in establishing firmly the foundations of these geometric approaches. The results of this work allow us to presume the importance of the global topological structures appearing in thermodynamics. Indeed, the fact that global locally-radial vector fields and extensive structures are equivalent might restrict the topology of a manifold of states \textit{a priori}. For instance, if a two-dimensional manifold of states has a~nowhere-vanishing~$\rho$, then it cannot be a sphere. The physical consequences that the topological properties of~$M$ might yield has not received any attention so far.

An important noninvariant feature of thermodynamic potentials that was disregarded in this work is convexity. The reason to overlook it is that this notion requires further geometric or algebraic structure. An attempt to define it in a coordinate-free fashion could take into account that Ruppeiner geometry and geometrothermodynamics endow $M$ with a (degenerate) metric tensor, which might allow for a definition similar to that of convex functions on Riemannian manifolds. Another possibility, however, is using the integral curves of $\rho$ as is done to define geodesically convex functions. Addressing this question is necessary for a coordinate-free statement of the principle of maximum entropy.

Together with~\cite{belgiorno2010a}, the geometric definition of extensivity that was presented here sheds new light on the geometry of thermodynamics. If we consider that a thermodynamic system is specified by the triad $(M,\vartheta,\rho)$, where $\vartheta$ and $\rho$ satisfy $\vartheta\wedge\ed\vartheta=0$, $\Lie{\rho}{\vartheta}=\vartheta$, and $\vartheta(\rho)\neq0$, then fundamental equations may be regarded as geometric equations involving only the geometric structure of~$M$, according to Belgiorno's equation (equation~\eqref{eq:Belgiorno}). Thus, establishing \textit{equations of motion} for the fields $\vartheta$ and $\rho$ would yield a non-phenomenological approach to macroscopic thermodynamics, written in a language that is common to other geometric physical theories. This approach would \textit{yield} the extensive structure and the heat form of each system, so that resorting to previously known ones would be unnecessary (cf.\ Example~\ref{Example:ig}).

According to the paragraph above, entropy is a distinguished potential, as it arises from the geometric structure on $M$. This privileges Ruppeiner's metric tensor over all other Hessian metrics that may be defined on $M$ by Hessians of thermodynamic potentials \cite{bravetti2014b}. It is important to point out that the latter are particular examples of information geometries \cite{brody2009a, crooks2007a}. The role of extensive structures in this broader context is still unknown.

Nothing was mentioned about extensivity in the thermodynamic phase space $P$. The reason is that the corresponding structure is not an extensive one. Indeed, let us suppose that $\big(P,\big(w,q^1,\ldots,q^n,p_1,\ldots,p_n\big)\big)$ is a global Darboux chart. We define a Legendre submanifold $\imath\colon M\hookrightarrow P$ by $\varPhi:=\imath^*w$, $x^i:=\imath^*q^i$, and $\imath^*p_i:=\partial_i\varPhi$, for every $i\in\{1,\ldots,n\}$. As customary, we assume that $x^1,\ldots,x^n$ are global extensive variables on $M$ and that $\varPhi$ satisfies $\Lie{\rho}{\varPhi}=\beta\varPhi$, for some real number $\beta$. The vector field $\rho:=x^i\partial_i$ defines an extensive structure on~$M$. A straightforward computation yields $\Der{\imath}{}\rho=\sigma\circ\imath$, where $\sigma:=\beta w\partial w+q^i\partial/\partial q^i+(\beta-1)p_j\partial/\partial p_j$. The latter is potentially useful to induce a notion of extensivity on~$P$, because we recover therefrom that $w$ and the functions $q^i$ are extensive variables (of different degrees) on $P$. Furthermore, $\Lie{\sigma}{\varTheta}=\beta\varTheta$. Like we mentioned before, this \textit{extensive structure} on $P$ does not coincide with the notion of extensivity that we presented in this paper, since it is not induced by a locally-radial vector field (cf.~\cite{schaft2018a}). This definition relies rather on the existence of global Darboux coordinates, which should be avoided to allow for less trivial topological structures on $P$. Moreover, the vector field $\sigma$ is determined modulo $\ker\Der{\imath}{}$. At this point, there is no straightforward feature of $\sigma$ that may help us characterize it in a coordinate-free fashion. This holds also in the case when $P$ is considered to be the matrix Lie group~$H_n$ of~\cite{preston2008a}. For instance, $\sigma$ does not belong to the corresponding Lie algebra, as follows upon observing that it vanishes at the identity. Hence, a coordinate-free notion of extensivity in the contact-geometric setting is a nontrivial task worth addressing.

There is a more general concept of extensivity that includes all known thermodynamic systems. It allows for the possibility of having different degrees of extensivity for each extensive variable in an extensive chart, and is particularly important in the case of Kerr--Newman black holes \cite{belgiorno2003a, quevedo2017a}. The ideas that we have presented in this work might help to describe this more general notion under a coordinate-free approach.

\subsection*{Acknowledgements}

The author wishes to thank Gerardo F.~Torres del Castillo, Merced Montesinos, Hernando Quevedo, and Alessandro Bravetti for their valuable comments regarding this work. The referees are acknowledged for helping to substantially improve this manuscript with their reports. This work was financially supported by VIEP, BUAP.

\pdfbookmark[1]{References}{ref}
\LastPageEnding

\end{document}